\documentclass{amsart}

\usepackage{graphicx}
\usepackage{hyperref}
\usepackage[all]{xy}
\usepackage[mathscr]{euscript}
\newtheorem{theorem}{Theorem}[section]
\newtheorem{lemma}[theorem]{Lemma}

\theoremstyle{definition}
\newtheorem{definition}[theorem]{Definition}
\newtheorem{example}[theorem]{Example}

\theoremstyle{remark}
\newtheorem{remark}[theorem]{Remark}

\numberwithin{equation}{section}

\newcounter{stepnum}
\newcommand{\step}{%
  \par
  \refstepcounter{stepnum}%
  \textbf{Step \arabic{stepnum}}.\enspace\ignorespaces
}

\begin{document}

\title{Solutions of the ${\rm SU}(n+1)$ Toda system from meromorphic functions}

\author{Yiqian Shi}
\address{CAS Wu Wen-Tsun Key Laboratory of Mathematics and  School of Mathematical \newline \indent Sciences, University of Science and Technology of China, Hefei 230026 China}
\email{yqshi@ustc.edu.cn}

\author{Chunhui Wei$^\dagger$}
\address{School of Gifted Young, University of Science and Technology of China\newline \indent Hefei 230026 China
}
\email{xclw3399@mail.ustc.edu.cn}
\thanks{$^\dagger$C.W. is the corresponding author.}

\author{Bin Xu}
\address{CAS Wu Wen-Tsun Key Laboratory of Mathematics and  School of Mathematical \newline \indent Sciences, University of Science and Technology of China, Hefei 230026 China}
\email{bxu@ustc.edu.cn}

\date{}

\dedicatory{}

\keywords{${\rm SU}(n+1)$ Toda system, meromorphic function, rational normal map}

\begin{abstract}
We consider the ${\rm SU}(n+1)$ Toda system on a simply connected domain $\Omega$ in ${\Bbb C}$, the $n=1$ case of which coincides with the Liouville equation $\Delta u+8e^u=0$. A classical result by Liouville
says that a solution of this equation on $\Omega$ can be represented by some non-degenerate meromorphic function on $\Omega$.
We construct a family of solutions parameterized by ${\rm PSL}(n+1,\,{\Bbb C})/{\rm PSU}(n+1)$ for the ${\rm SU}(n+1)$ Toda system from such a meromorphic function on $\Omega$, which generalizes the result of Liouville.
As an application, we find a new class of solvable ${\rm SU}(n+1)$ Toda systems with singular sources via cone spherical metrics on compact Riemann surfaces.
\end{abstract}

\maketitle


\section{Introduction to Toda system}\label{intro}
Let $\Omega$ be a simply connected domain in ${\Bbb C}$ and $n$ a positive integer.
Consider the following open ${\rm SU}(n+1)$ Toda system
\begin{equation}
	-\frac{1}{4}\triangle u_i=\sum_{j=1}^na_{ij}e^{u_j},\quad 1\leq i\leq n,\quad\text{on}\quad\Omega,\label{1.1}
\end{equation}
where $\displaystyle{\Delta=4\frac{\partial^2}{\partial z\partial\bar z}=\frac{\partial^2}{\partial x^2}+\frac{\partial^2}{\partial y^2}}$ is the Laplace operator in ${\Bbb C}={\Bbb R}^2$ and
\begin{equation}
(a_{ij})=
\begin{pmatrix}
2 & -1 & 0 &\cdots&\cdots& 0\\
-1& 2  & -1& 0 &\cdots& 0\\
0&-1& 2  & -1&\cdots& 0\\
\vdots&\vdots&\vdots&\vdots&\vdots&\vdots\\
0 &\cdots& 0 & -1 & 2 &-1\\
0 &\cdots&\cdots& 0 & -1 & 2
\end{pmatrix}\label{1.2}
\end{equation}
is the Cartan matrix of ${\rm SU}(n+1)$.

We quickly review the correspondence between solutions of (\ref{1.1}) and totally unramified
holomorphic curves $\Omega\to{\Bbb P}^n$, whose definition refers to \cite[p.270]{GH:1994}).
Let $(u_1,\cdots,u_n)$ be an $n$-tuple of real-valued smooth function on $\Omega$
and define the $(n+1)$-tuple $(w_0,\cdots, w_n)$ of functions on $\Omega$ by
\begin{equation}
\begin{cases}
	w_0:=-\frac{\sum_{i=1}^n(n-i+1)u_i}{2(n+1)}\\
	w_i:=w_0+\frac{1}{2}\sum_{j=1}^iu_j,\quad 1\leq i\leq n.\label{1.3}
\end{cases}
\end{equation}
Then $(u_1,\cdots,u_n)$ solves the
SU(n+1) Toda system (\ref{1.1}) if and only if there holds the Maurer-Cartan equation (i.e. the integrability condition)
$$\mathcal{U}_z-\mathcal{V}_{\bar{z}}=[\mathcal{U},\mathcal{V}],$$
where
\begin{equation}
\mathcal{U}=
\begin{pmatrix}
(w_0)_z&&&\\
&(w_1)_z&&\\
&&\ddots&\\
&&&(w_n)_z
\end{pmatrix}
+\begin{pmatrix}
0&&&\\
e^{w_1-w_0}&0&&\\
&\ddots&\ddots&\\
&&e^{w_n-w_{n-1}}&0
\end{pmatrix}
\nonumber
\end{equation}
\begin{equation}
\mathcal{V}=
-\begin{pmatrix}
(w_0)_{\bar{z}}&&&\\
&(w_1)_{\bar{z}}&&\\
&&\ddots&\\
&&&(w_n)_{\bar{z}}
\end{pmatrix}
-\begin{pmatrix}
0&e^{w_1-w_0}&&\\
&0&\ddots&\\
&&\ddots&e^{w_n-w_{n-1}}\\
&&&0
\end{pmatrix}
\nonumber
\end{equation}
By using the Frobenius theorem and the analytic-continuation-like argument,
we obtain for the solution $(u_1,\cdots, u_n)$ of (\ref{1.1}) the so-called {\it Toda map}
$\phi:\Omega\to {\rm SU}(n+1)$ which satisfies
\[\phi^{-1}{\rm d}\phi=\mathcal{U}{\rm d}z+\mathcal{V}{\rm d}\bar{z},\]
and is unique up to a multiple in ${\rm SU}(n+1)$ (\cite[Section 3.1]{JW:2002}).
We cite here the correspondence between solutions of (\ref{1.1}) and totally unramified holomorphic curves $\Omega\to \mathbb{P}^n$ modulo unitary transformations on ${\Bbb P}^n$, which was established by Jost-Wang \cite[Section 3]{JW:2002}.

\begin{theorem}{\rm{\cite[Section 3]{JW:2002}}}\label{thm1.1} Let $\{e_0,\cdots, e_n\}$ be the standard basis in the Hermitian inner product space
${\Bbb C}^{n+1}$, which induces the canonical Hermitian inner products on
$\Lambda^k{\Bbb C}^{n+1}$ for $0\leq k\leq n$.
\begin{enumerate}
\item Let $f:\Omega\to\mathbb{CP}^n$ be a totally unramified  holomorphic curve. Then $f$ has a lifting $\hat{f}:\Omega\to\mathbb{C}^{n+1}$, unique up to a constant multiple lying in
     $\{\xi\in{\Bbb C}:\xi^{n+1}=1\}$, such that $\hat{f}\wedge \hat{f}^\prime\wedge \hat{f}^{(2)}\wedge\cdots\wedge \hat{f}^{(n)}=e_0\wedge\cdots\wedge e_n$, i.e. the Wronskian of
     of $\hat{f}$ equals one identically. Define $\Lambda_k(\hat{f}):=\hat{f}^{(0)}\wedge\cdots\wedge \hat{f}^{(k)}$ for $k=0,\cdots,n-1$ and
 \begin{equation}
u_i:=-\sum_{j=1}^na_{ij}\log\| \Lambda_{j-1}(\hat{f})\|^2(1\leq i\leq n)\label{1.4},
\end{equation}
where $\| \Lambda_k\big(\hat{f}(z)\big)\|^2:\Omega\to (0,\,+\infty)$ is the norm  of $\Lambda_k\big(\hat{f}(z)\big)$ in $\Lambda^k{\Bbb C}^{n+1}$ for each $k$ in $\{0,1,\cdots, n-1\}$. Then $(u_1,\cdots, u_n)$ is a solution of (\ref{1.1}).

\item Suppose that $\phi:\Omega\to {\rm SU}(n+1)$ is a Toda map associated to a solution $(u_1,\cdots,u_n)$ of (\ref{1.1}). Defining an $(n+1)$-tuple $(\hat{f}_0,\cdots,\hat{f}_n)$ of $\mathbb{C}^{n+1}$-valued functions on $\Omega$ by
$$
(\hat{f}_0,\cdots,\hat{f}_n)=\phi\cdot
\begin{pmatrix}
e^{w_0}&&&\\
&e^{w_1}&&\\
&&\ddots&\\
&&&e^{w_n}
\end{pmatrix},
$$
we find that $\hat{f}_0$ is a holomorphic ${\Bbb C}^{n+1}$-valued function on $\Omega$
which  satisfies $\hat{f}_0\wedge \hat{f}_0^\prime\wedge \hat{f}_0^{(2)}\wedge\cdots\wedge \hat{f}_0^{(n)}=e_0\wedge\cdots\wedge e_n$ and induces a totally unramified holomorphic curve $f_0:\Omega\to {\Bbb P}^n$. Moreover, $(u_1,\cdots,u_n)$ coincides with
the solution of $(\ref{1.1})$ constructed from the holomorphic curve $f_0$ by (\ref{1.4}).

\end{enumerate}
\end{theorem}

\begin{definition}
We call $f_0:\Omega\to {\Bbb P}^n$ in Theorem \ref{thm1.1} (2) the (totally unramified) {\it holomorphic curve associated to the solution} $(u_1,\cdots, u_n)$ of the ${\rm SU}(n+1)$ Toda system.
Note that such an associated curve is unique up to a unitary transformation on ${\Bbb P}^n$ (\cite[(4.12)]{Griffiths:1974}).
In particular, if $n=1$, then the holomorphic curve $f_0:\Omega\to {\Bbb P}^1$ associated to
a solution $u_1$ of the Liouville equation
$\Delta u_1+8e^{u_1}=0$ turns out to be a non-degenerate meromorphic function on $\Omega$.
\end{definition}

The organization of this manuscript goes as follows. We state the main results (Theorems \ref{thm:rn} and \ref{m.s.}) in Section 2 and prove the key lemma of Theorem \ref{thm:rn} in Section 3. In Section 4 we 
introduce the readers carefully the long journey which motivates us Theorem \ref{thm:rn}. We prove 
Theorem \ref{m.s.} in the last section.

\section{Main results}
Motivated by the following result, we investigate a particular class of solutions of ${\rm SU}(n+1)$ Toda system (\ref{1.1}) whose
associated curves are obtained by compositions of linear transformations on ${\Bbb P}^n$,  rational normal maps, and non-degenerate meromorphic functions on $\Omega$. Moreover, we find an analytic formula for such solutions.

\begin{theorem}
\label{thm:rn}
Let $u:\Omega\to {\Bbb R}$ be a solution of the Liouville equation $\Delta u+8e^u=0$ on $\Omega$
and $v:\Omega\to\mathbb{P}^1$ a curve associated to $u$. Setting
$$c_i:=\ln\, i+\ln\,(n+1-i)\quad {\rm for}\quad i\in\{1,2,\cdots, n\},$$
we find
that $(u_1,\cdots, u_n):=(u+c_1,\cdots, u+c_n)$ solves (\ref{1.1}) and
$f=r_n\circ v$ is an associated curve of it, where $r_n:{\Bbb P}^1\to {\Bbb P}^n$ is a rational normal map
with form
$$r_n:[z_0,z_1]\mapsto\left[\sqrt{\frac{1}{n!}}\,z_0^n,\sqrt{\frac{1}{(n-1)!1!}}\,z_0^{n-1}z_1,
\cdots,\sqrt{\frac{1}{n!}}\,z_1^n\right].$$
\end{theorem}
\begin{proof}
A direct computation shows that $(u_1,\cdots, u_n):=(u+c_1,\cdots, u+c_n)$ does solve (\ref{1.1}).
Denote by $v=[v_0,\, v_1]$ the curve $v:\Omega\to {\Bbb P}^1$ associated to $u$ such that $v_0(z)v'_1(z)-v_1(z)v'_0(z)\equiv 1$ on $\Omega$. By Lemma \ref{lm:wr1}, the canonical lifting
$$\hat{f}=\left(\sqrt{\frac{1}{n!}}\,v_0^n:\sqrt{\frac{1}{(n-1)!1!}}\, v_0^{n-1}v_1:
\cdots:\sqrt{\frac{1}{n!}}\, v_1^n\right):\Omega\to {\Bbb C}^{n+1}$$
of the curve
$$f=r_n\circ v=\left[\sqrt{\frac{1}{n!}}\,v_0^n:\sqrt{\frac{1}{(n-1)!1!}}\, v_0^{n-1}v_1:
\cdots:\sqrt{\frac{1}{n!}}\, v_1^n\right]:\Omega\to {\Bbb P}^n$$
has Wronskian $\equiv 1$  i.e. $\hat{f}\wedge \hat{f}'\wedge\cdots\wedge \hat{f}^{(n)}\equiv e_0\wedge\cdots\wedge e_n$ on $\Omega$(It also means that $f$ is totally unramified).\par
It suffices to check that $u+\ln n$ equals
the first component $u_1$ of
solution $(u_1,\cdots, u_n)$ from the lifting $\hat f$ of the curve $f$. Actually, we have
\begin{equation}
\begin{aligned}
 u_1&=\log\left(\frac{\| \Lambda_1(\hat{f})\|^2}{\|\hat{f}\|^4}\right)\\
	&=\log\left(\frac{\partial^2}{\partial z\partial \bar{z}}\log\| \hat{f}\|^2\right)\\
	&=\log\left(\frac{\partial^2}{\partial z\partial \bar{z}}\log\Big(n!(|v_0|^2+|v_1|^2)^n\Big)\right)\\
	&=\log\left(\frac{\partial^2}{\partial z\partial \bar{z}}\log(|v_0|^2+|v_1|^2)\right)+\ln n\\
	&=u+\ln n,
\end{aligned}
\end{equation}
where we use the infinitesimal Pl\" ucker formula (\cite[p.269]{GH:1994}) in the second equality.
\end{proof}
\begin{remark}
This proof depends on the fact that we know the curve associated to the solution $(u+c_i)_{1\leq i\leq n}$ a priori. However, it was a little bit difficult for us to reach it. We shall narrate this process in Section 4.
\end{remark}

Beside $r_n\circ v:\Omega\to {\Bbb P}^n$ in Theorem \ref{thm:rn}, we further come up with more totally unramified curves
with form $\mathcal{A}\circ r_n\circ v:\Omega\to {\Bbb P}^n$, where $\mathcal{A}\in{\rm Aut}({\Bbb P}^n)={\rm PSL}(n+1,\,{\Bbb C})$. These curves may produce more solutions of (\ref{1.1}) by the process in Item (1) of Theorem \ref{thm1.1}. Recall that curves  $\mathcal{U}\circ r_n\circ v:\Omega\to {\Bbb P}^n$ give the same
solution as $r_n\circ v:\Omega\to {\Bbb P}^n$ where $\mathcal{U}$'s are unitary transformations on ${\Bbb P}^n$.
Recalling the Gram-Schmidt orthogonalization process, we could
only consider upper triangular linear transformations whose eigenvalues are all positive. Such transformations form a real Lie group of dimension $n^2+2n$, denoted by
$R_n\cong {\rm PSL}(n+1,\,{\Bbb C})/{\rm PSU}(n+1)$. Finally, we arrive at the definition of {\it reduced solution}
of ${\rm SU}(n+1)$ Toda system from the Liouville equation $\Delta u+8e^u=0$ on $\Omega$.

\begin{definition}
We call a solution $(u_1,\cdots, u_n)$ of $SU(n+1)$ Toda system (\ref{1.1}) {\it reduced} from the Liouville equation $\Delta u+8e^u=0$ if and only if some associated curve of $(u_1,\cdots, u_n)$ has form
$\mathcal{R}\circ r_n\circ v$, where $v:\Omega\to {\Bbb P}^1$ is a non-degenerate meromorphic function
and $\mathcal{R}$ lies in $R_n$.
\end{definition}

Here is the main result of this manuscript.

\begin{theorem}\label{m.s.}
A solution $(u_1,\cdots, u_n)$ of ${\rm SU}(n+1)$ Toda system $(\ref{1.1})$ is reduced from the Liouville equation if and only if there is a non-degenerate meromorphic function $f$ on $\Omega$ such that
$$u_1=\log\frac{\sum_{i,j,l,m=1}^{n+1}R_{i,j,l,m}\big(f^{i-1}\bar{f}^{j-1}(f^{\prime})^{ l-1}(\bar{f^\prime})^{m-1}-f^{i-1}(\bar{f^\prime})^{j-1}(f^{\prime})^{l-1}\bar{f}^{m-1}\big)}
{\Big(\sum_{i=1}^{n+1}\sum_{j=1}^{n+1}\sum_{k=1}^{min\{i,j\}}R_{k,i}\bar{R}_{k,j}f^{i-1}\bar{f}^{j-1}\Big)^2},$$
where $R_{i,j}\in\mathbb{C}$ for all $1\leq i\leq j\leq n$ satisfy $\Pi_{m=0}^n(m!R_{m,m})=1$, $R_{m,m}>0$ for all $0\leq m\leq n$,  and
$$R_{i,j,l,m}=\left(\sum_{k=1}^{min\{i,j\}}R_{k,i}\bar{R}_{k,j}\right)
\left(\sum_{k=1}^{min\{l,m\}}R_{k,l}\bar{R}_{k,m}\right).$$
\end{theorem}

As an application, we find a new class of solvable ${\rm SU}(n+1)$ Toda systems with finitely many singular sources on compact Riemann surfaces in the following:

\begin{theorem}
{\it
Suppose that there exists a cone spherical metric $g$ representing the real divisor $D=\sum_{j=1}\,\gamma_j[P_j]$, $\gamma_j>-1$ for all $1\leq j\leq n$, on a compact Riemann surface $X$ {\rm (\cite[Section 2]{LSX:2021})}. Then for each positive integer $n$, we could solve the ${\rm SU}(n+1)$ Toda system on $X$ with the singular source being the $n$-tuple of real divisors $(D,\cdots, D)$. In other words, the system has form
\[-\frac{\partial^2 u_i}{\partial z\partial \bar z}=\sum_{j=1}^n\,a_{ij}e^{u_j}+\sum_{P_k\in U}\, \pi\gamma_k\delta_{P_k} \quad \text{for all}\quad 1\leq i\leq n\]
at each complex coordinate chart $(U,\,z)$ of $X$. Moreover, the solution space of such Toda systems contains a subspace consisting of
those solutions with associated curves being compositions of rational normal curves and developing maps of
$g$.}
\end{theorem}

We shall establish a complex analytical framework for solutions of ${\rm SU}(n+1)$ Toda systems with singular sources on Riemann surfaces and prove this theorem in another manuscript.

\section{The key lemma of Theorem \ref{thm:rn}}

We prove the key lemma as follows for Theorem \ref{thm:rn} in this section.

\begin{lemma}\label{lm:wr1}
Let $v_0,v_1:\Omega\to\mathbb{C}$ be holomorphic functions such that $$v_0(z)v'_1(z)-v_1(z)v'_0(z)\equiv 1$$ on $\Omega$. Then the canonical lifting
$f=\left(\sqrt{\frac{1}{n!}}\,v_0^n,\sqrt{\frac{1}{(n-1)!1!}}\, v_0^{n-1}v_1,
\cdots,\sqrt{\frac{1}{n!}}\, v_1^n\right):\Omega\to {\Bbb C}^{n+1}$
has Wronskian $\equiv 1$.
\end{lemma}
We need to prepare two other lemmas for the proof of this one.

\begin{lemma}\label{lm:wr2}
Let $f=(f_0,\cdots,f_n):\Omega\to\mathbb{C}^{n+1}$ and $v:\Omega\to\mathbb{C}$ be holomorphic. Denote  
$v\cdot f:=(vf_0,\cdots,vf_n)$. Then $\Lambda_n(v\cdot f)=v^{n+1}\Lambda_n(f)$.
\end{lemma}
\begin{proof}
Since $(v\cdot f)^{(k)}=\sum_{i=0}^{k}v^{(i)}\cdot f^{(k-i)}$, we have
\begin{equation}
\begin{aligned}
\Lambda_n(v\cdot f)&=(v\cdot f) \wedge (v\cdot f)^{(1)} \wedge \cdots \wedge (v\cdot f)^{(n)}\\
				   &=(v\cdot f) \wedge (v^{(1)}\cdot f+v\cdot f^{(1)}) \wedge \cdots \wedge (\sum_{i=0}^{n}v^{(i)}\cdot f^{(k-i)})\\
				   &=(v\cdot f)\wedge (v\cdot f^{(1)})\wedge\cdots\wedge (v\cdot f^{(n)})(\text{because } v^{(i)}\cdot f^{(k)}\wedge v^{(j)}\cdot f^{(k)}=0,\forall i,j,k)\\
				   &=v^{n+1}f\wedge f^{(1)}\wedge\cdots\wedge f^{(n)}\\
				   &=v^{n+1}\Lambda_n(f)
\end{aligned}
\nonumber
\end{equation}
\end{proof}
\begin{lemma}\label{lm:wr3}
Let $v:\Omega\to\mathbb{C}$ be non-degenerate meromorphic function and $f=\big( 1,\frac{1}{1!}v,\cdots,\frac{1}{n!}v^n\big):\Omega\to\mathbb{C}^{n+1}$. Then $\Lambda_n(f)=(v^\prime)^{\frac{n(n+1)}{2}}$.
\end{lemma}
\begin{proof}
We prove it by induction.
\begin{enumerate}
\item Case $n=1$ is easy.\\

\item Suppose that $n\geq 2$ and for all $1\leq k\leq n-1$, we have
$$\Lambda_k\left( 1,\frac{1}{1!}v,\cdots,\frac{1}{k!}v^k\right)=(v^\prime)^{\frac{k(k+1)}{2}}.$$
Then
\begin{equation}
\begin{aligned}
\Lambda_n(f)&=\Lambda_n\left( 1,\frac{1}{1!}v,\cdots,\frac{1}{n!}v^n\right)\\
			&=\Lambda_{n-1}\left(\frac{1}{1!}v^\prime,\cdots,\frac{1}{n!}nv^{n-1}v^\prime\right)\\
			&=(v^\prime)^{n+1}\Lambda_{n-1}\left(\frac{1}{1!},\cdots,\frac{1}{(n-1)!}v^{n-1}\right)(\text{by Lemma }\ref{lm:wr2})\\
			&=(v^\prime)^{n+1}(v^\prime)^{\frac{n(n-1)}{2}}\\
			&=(v^\prime)^{\frac{n(n+1)}{2}}
\end{aligned}
\nonumber
\end{equation}
\end{enumerate}
\end{proof}
\begin{proof}[Proof of Lemma \ref{lm:wr1}]
Let $v=v_1/v_0$. Then we have $v^\prime=\frac{1}{v_0^2}$ and
\begin{equation}
\begin{aligned}
\Lambda_n(f)&=\Lambda_n\left(\sqrt{\frac{1}{n!}}\,v_0^n:\sqrt{\frac{1}{(n-1)!1!}}\, v_0^{n-1}v_1:
\cdots:\sqrt{\frac{1}{n!}}\, v_1^n\right)\\
			&=\Lambda_n\left(\sqrt{\frac{1}{n!}}\,(v^\prime)^{-\frac{n}{2}}:\sqrt{\frac{1}{(n-1)!1!}}\, (v^\prime)^{-\frac{n}{2}}v:
\cdots:\sqrt{\frac{1}{n!}}\, (v^\prime)^{-\frac{n}{2}}v^n\right)\\
			&=(v^\prime)^{-\frac{n(n+1)}{2}}\Lambda_n\left(\sqrt{\frac{1}{n!}}\,:\sqrt{\frac{1}{(n-1)!1!}}\, v:
\cdots:\sqrt{\frac{1}{n!}}\, v^n\right)(\text{by Lemma }\ref{lm:wr2})\\
			&=(v^\prime)^{-\frac{n(n+1)}{2}}(v^\prime)^{\frac{n(n+1)}{2}}(\text{by Lemma }\ref{lm:wr3})\\
			&=1.
\end{aligned}
\nonumber
\end{equation}
\end{proof}

\section{why rational normal map}
Let $u:\Omega\to {\Bbb R}$ be a solution of the Liouville equation
\begin{equation}
\label{2.1}
\triangle u+8e^u=0\quad {\rm on}\quad \Omega
\end{equation}
and $v=[v_0,\, v_1]:\Omega\to {\Bbb P}^1$ a curve associated to $u$. Since the Wronskian $v_0v'_1-v'_0v_1$ of
the lifting $(v_0,\,v_1)$ of $v$ equals $1$ on $\Omega$, we have $$u=-\ln\,\| v\|^4=-\ln\,\Big(\| v_0\|^2+\| v_1\|^2\Big)^2.$$ Theorem \ref{thm:rn} says that
$$(u_i)_{1\leq i\leq n}:=(u+c_i)_{1\leq i\leq n}=\Big(u+\ln\, i +\ln\,(n+1-i) \Big)_{1\leq i\leq n}$$ solves (\ref{1.1}) and has $r_n\circ v:\Omega\to {\Bbb P}^n$ as its associated curve. However, it seems that the difficulty hides in
how we came up with the rational normal map $r_n$.
In this section, we discuss our mental journey of overcoming this difficulty.

By using the substitution (\ref{1.4}), we obtain
\begin{eqnarray*}
w_k&=&\frac{n-2k}{2}\log\| v\|^2+\frac{1}{2}\sum_{i=1}^kc_i-\frac{1}{2}C,\,\,
{\rm where}\,\, C=\frac{1}{n+1}\sum_{i=1}^n(n-i+1)c_i,\,\, {\rm and} \\
\frac{\partial w_k}{\partial z}&=&\frac{n-2k}{2}\frac{(v_z,v)}{\| v\|^2}.
\end{eqnarray*}
To understand how an associated curve to $(u_i)_{1\leq i\leq n}$ is relevant to curve $v$, we need to use the following result in  Jost-Wang \cite{JW:2002}.

\begin{theorem}{\rm{\cite[section 3]{JW:2002}}}\label{thm2.1}
Let $\phi$ be a Toda map of $(u_i)_{1\leq i\leq n}$ and $(\hat{f}_0,\cdots,\hat{f}_n):=\phi\cdot \Big({\rm diag}(e^{w_i})_{0\leq i\leq n}\Big)$. Then $\hat{f}_0$ is the lift of a holomorphic curve $f_0$ associated to $(u_i)_{1\leq i\leq n}$  and $\hat{f}_i$ satisfies the following equations:
\begin{equation}\label{2.2}
	\begin{cases}
			\begin{aligned}
				\hat{f}_{k+1} &= \frac{\partial \hat{f}_k}{\partial z}-2(w_k)_z\hat{f_k}\\
							  &= \frac{\partial \hat{f}_k}{\partial z}+(2k-n)\frac{(v_z,v)}{\| v\|^2}\hat{f}_k \,\,{\rm for\ all}\,\, 0\leq k< n;\\
				\hat{f}_{n+1} &= 0
			\end{aligned}
	\end{cases}
		\end{equation}
	and	
		\begin{equation}\label{2.3}
			\begin{cases}
			\begin{aligned}
				\frac{\partial \hat{f}_k}{\partial \bar{z}} &= -e^{u_k}\hat{f}_{k-1}=-\frac{e^{c_k}}{\| v\|^4}\hat{f}_{k-1}
 \,\,{\rm for\ all}\,\, 0< k\leq n.\\
											   \frac{\hat{f}_{0}}{\partial\bar{z}} &= 0
			\end{aligned}
		\end{cases}
		\end{equation}\par
\end{theorem}
By using equation (\ref{2.2}), we could obtain that $\hat{f}_0$ satisfies the following linear ODE of $(n+1)$th order.

\begin{lemma}\label{lem2.2}
Let $P$ be the differential operator acting on complex-valued functions on $\Omega$ defined by
$\displaystyle{Pf:=\frac{\partial}{\partial z}\biggl(\| v\|^4 f\biggr)}$.
Then we have
\begin{equation}\label{2.4}
P^{n+1}\Big(\| v\|^{-2(n+2)}\hat{f}_0\Big)=0.
\end{equation}
That is, all the $(n+1)$ components of $\hat{f}_0$ satisfy this ODE.
\end{lemma}
\begin{proof}
At first we rewrite (\ref{2.2}) into the following equations:
\begin{equation}
			\begin{aligned}
				\hat{f}_{k+1} &= \| v\|^{2(n-2k)}\frac{\partial}{\partial z}\bigg(\| v\|^{2(2k-n)}\hat{f}_k\bigg)\,\,{\rm for\ all}\,\,  0\leq k\leq n;\\
				\hat{f}_{n+1} &= 0
			\end{aligned}
		\end{equation}
Denoting $\varphi_k:=\| v\|^{2(2k-2-n)}\hat{f}_k$,
we obtain
		$$\varphi_{k+1}=\frac{\partial}{\partial z}\bigg(\| v\|^4\varphi_k\bigg)=P\varphi_k\,\,{\rm for\ all}\,\,  0\leq k\leq n.$$
Hence we find $\phi_{n+1}=P\phi_n=\cdots=P^{n+1}\phi_0$ and
$$P^{n+1}\Big(\| v\|^{-2(n+2)}\hat{f}_0\Big)=0.$$
\end{proof}

To solve the equation
\begin{equation}
\label{equ:P}
P^{n+1}\Big(\| v\|^{-2(n+2)}f\Big)=0,
\end{equation}
we need to quote some preliminary results about linear ODEs on a {\it differential field}, whose definition is given in the following:

\begin{definition}{\rm{\cite[Section 2]{Dale:2016}}}
A {\it differential field} $(K,\partial)$ is a field equipped with a map
$\partial: K\to K$, called a {\it derivation} over $K$,
such that for all $x,\,y\in K$ there holds
\begin{equation}
	\begin{aligned}
		&\partial(x+y)=\partial x+\partial y\\
		&\partial(xy)=x\partial y+y\partial x
	\end{aligned}
\nonumber
\end{equation}
Call $C_K:=\{x\in K|\partial x=0\}$ the {\it constant field} of $(K,\, \partial)$.
\end{definition}
\begin{lemma}{\rm{\cite[Remark 3.15]{Dale:2016}}}\label{lem 2.3}
Let $(K,\partial)$ be a differential field and
$$Lx:=\sum_{i=1}^n\, a_i\partial^i x=0$$ be a homogeneous linear  ODE of $n$th order on $K$, where
$a_i\in K$ for all $1\leq i\leq n$ and $a_n\not=0$.
Then the solution space $V:=\{x\in K:Lx=0\}$ has dimension not exceeding $n$ over the constant field $C_K$.
\end{lemma}

\begin{lemma}{\rm{\cite[Proposition 3.14]{Dale:2016}}}
\label{lem:indep}
The $n$ elements $x_1,\cdots,x_n$ in $K$ are  linearly dependent over the constant field $C_K$ if and only if $$W(x_1,\cdots,x_n):=\det\, \Big(\partial^j x_i\Big)_{1\leq i\leq n,\, 0\leq j\leq n-1}=0.$$
\end{lemma}

We could think of equation \eqref{equ:P} as a homogeneous linear ODE of $(n+1)$th order in a differetial field in the following:

\begin{example}
Let $R_1$ be the ring generated by real analytic complex-valued functions on $\Omega$ and
$R$ its subring containing all the holomorphic functions and all the anti-holomorphic ones on $\Omega$.
Recall that the zero set of a nonconstant real analytic function on $\Omega$ has measure zero
(\cite{Mit:2020}). Hence both $R_1$ and $R$ are integral rings.

Let $K:=\text{Frac}(R)$ be the fractional field  of $R$. Then $\left(K,\partial_z:=\frac{\partial}{\partial z}\right)$ forms a differential field and $$C_K=\{f\in K|\partial_z\,f=0\}=\text{Frac}\big(\left\{\text{anti-holomorphic function on}\,\Omega\right\}\big).$$
Then the equation (\ref{equ:P}) is a homogeneous linear ODE of $(n+1)$th on $K$. Let $V$ be the linear space of solutions of (\ref{equ:P}) on $K$, and $H$ the linear space consisting of holomorphic functions in $V$.
By Lemma \ref{lem 2.3},  $V$ has dimension not exceeding $(n+1)$ over $C_K$.
\end{example}

\begin{lemma}
$\dim_{\mathbb{C}}H\leq n+1$.
\end{lemma}
\begin{proof}
Choose ${\Bbb C}$-linearly independent elements $f_1,\cdots, f_k$ in $V$.
Then
$W(f_1,\cdots,f_k)$
nowhere vanishes on $\Omega$.
By Lemmas \ref{lem:indep} and \ref{lem 2.3}, $f_1,\cdots,f_k$ are linearly independent over $C_k$ and $k\leq n+1$.
\end{proof}
At last, we solve (\ref{equ:P}) in the following:

\begin{lemma}
$\{v_0^n,v_0^{n-1}v_1,\cdots,v_1^n\}$ is a base of $H$ over ${\Bbb C}$.
\end{lemma}
\begin{proof}
We divide the proof into two steps.
\step {\it For all $0\leq k\leq n$, $v_0^kv_1^{n-k}$ solve (\ref{2.4})}.\\
We prove it by induction.
\begin{enumerate}
\item By direct computation, we find that $v_0,\, v_1$ solve $P^2(\| v\|^{-6}\hat{f}_0)=0$.

\item Suppose that for all $0\leq k\leq n-1$, $P^n\big(\| v\|^{-2(n+1)}v_0^kv_1^{n-k-1}\big)=0$. Then we show that for all $0\leq k\leq n$, $P^{n+1}\big(\| v\|^{-2(n+2)}v_0^kv_1^{n-k}\big)=0$. Actually, we have
\begin{equation}
		\begin{aligned}
			 &P\big(\| v\|^{-2(n+2)}v_0^kv_1^{n-k}\big)\\
			=&\frac{\partial}{\partial z}\biggl(\| v\|^{-2n}v_0^kv_1^{n-k}\biggr)\\
			=&-n\| v\|^{-2(n+1)}v_0^kv_1^{n-k}(v_z,v)+k\| v\|^{-2n}v_0^{k-1}v_1^{n-k}v_0^\prime+(n-k)\| v\|^{-2n}v_0^kv_1^{n-k-1}v_1^\prime\\
			=&k\| v\|^{-2(n+1)}v_0^{k-1}v_1^{n-k}(v_0^\prime\| v\|^2-v_0(v_z,v))\\
			 &+(n-k)\| v\|^{-2(n+1)}v_0^kv_1^{n-k-1}(v_1^\prime\| v\|^2-v_1(v_z,v))\\
			=&-k\bar{v}_1\| v\|^{-2(n+1)}v_0^{k-1}v_1^{n-k}+(n-k)\bar{v}_0\| v\|^{-2(n+1)}v_0^kv_1^{n-k-1}\quad {\rm and}
		\end{aligned}
\nonumber
	\end{equation}
\begin{equation}
\begin{aligned}
&P^{n+1}(\| v\|^{-2(n+2)}v_0^kv_1^{n-k})\\
=&P^n(-k\bar{v}_1\| v\|^{-2(n+1)}v_0^{k-1}v_1^{n-k}+(n-k)\bar{v}_0\| v\|^{-2(n+1)}v_0^kv_1^{n-k-1})\\
=&-k\bar{v}_1P^n(\| v\|^{-2(n+1)}v_0^{k-1}v_1^{n-k})+(n-k)\bar{v}_0P^n(\| v\|^{-2(n+1)}v_0^kv_1^{n-k-1})\\
=&0
\end{aligned}
\nonumber
\end{equation}
\end{enumerate}

\step {\it $\{v_0^n,v_0^{n-1}v_1,\cdots,v_1^n\}$ is a base of $H$}. In fact,
we observe that $v_0^n,v_0^{n-1}v_1,\cdots,v_1^n$ are linearly independent over ${\Bbb C}$.
On the other hand, since $\dim_{\mathbb{C}}H\leq n+1$,
$H$ is spanned by $v_0^n,\, v_0^{n-1}v_1,\,\cdots,v_1^n$ over ${\Bbb C}$.
\end{proof}
This lemma tells us that $\hat{f}_0=\mathcal{A}\circ
\left[(\frac{1}{n!})^{\frac{1}{2}}v_0^n:(\frac{1}{(n-1)!1!})^{\frac{1}{2}}v_0^{n-1}v_1:
\cdots:(\frac{1}{n!})^{\frac{1}{2}}v_1^n\right]$, where $\mathcal{A}$  is a non-degenerate linear transformation on ${\Bbb C}^{n+1}$. Moreover, we could assume $\mathcal{A}\in R_{n+1}$ modulo unitary transformations.

\section{Proof of Theorem \ref{m.s.}}
\begin{proof}[Proof of Theorem \ref{m.s.}]
\quad\\
"ONLY IF" part:\par
Let the corresponding curve of $\{u_i\}$ be $\phi=\mathcal{A}\circ r_n\circ v$, where $R=(r_{i,j})_{n+1\times n+1}\in {\mathcal R}_{n+1}, v=[v_0:v_1]\in H_1$. Denote $f=\frac{v_1}{v_0}$ and $v_0(z)v'_1(z)-v_1(z)v'_0(z)\equiv 1$ on $\Omega$. Then $v=[(f^\prime)^{-\frac{1}{2}}:(f^\prime)^{-\frac{1}{2}}f]$, which means
$$\phi=(f^\prime)^{-\frac{n}{2}}(\frac{\sum_{i=0}^nr_{1,i+1}f^i}{(n!)^{1/2}},\frac{\sum_{i=1}^nr_{2,i+1}f^i}{(1!(n-1)!)^{1/2}},\cdots,\frac{r_{n+1,n+1}f^n}{(n!)^{1/2}}).$$\par
Let $R_{i,j}=r_{i,j}((i-1)!(n-i+1)!)^{-\frac{1}{2}}$. Then we have the following:
\begin{equation}
\begin{aligned}
\|\phi\|^2&=|f^\prime|^{-n}\sum_{i=1}^{n+1}|\sum_{j=i}^{n+1}R_{i,j}f^{j-1}|^2\\
						&=|f^\prime|^{-n}(\sum_{i=1}^{n+1}\sum_{j=1}^{n+1}\sum_{k=1}^{min\{i,j\}}R_{k,i}\bar{R}_{k,j}f^{i-1}\bar{f}^{j-1})
\end{aligned}
\nonumber
\end{equation}
Thus, we can get $u_1$.
\begin{equation}
\begin{aligned}
u_1 &=-2\log\|\Lambda_0(\phi)\|^2+\log\|\Lambda_1(\phi)\|^2\\
	&=\log(\frac{\partial^2}{\partial z\partial\bar{z}}log(\|\phi\|^2))\\
	&=\log(\frac{\partial^2}{\partial z\partial\bar{z}}log(|f^\prime|^{-n}(\sum_{i=1}^{n+1}\sum_{j=1}^{n+1}\sum_{k=1}^{min\{i,j\}}R_{k,i}\bar{R}_{k,j}f^{i-1}\bar{f}^{j-1})))\\
	&=\log\frac{\sum_{i,j,l,m=1}^{n+1}(\sum_{k=1}^{min\{i,j\}}R_{k,i}\bar{R}_{k,j})(\sum_{k=1}^{min\{l,m\}}R_{k,l}\bar{R}_{k,m})(f^{i-1}\bar{f}^{j-1}f^{\prime l-1}\bar{f^\prime}^{m-1}-f^{i-1}\bar{f^\prime}^{j-1}f^{\prime l-1}\bar{f}^{m-1})}{(\sum_{i=1}^{n+1}\sum_{j=1}^{n+1}\sum_{k=1}^{min\{i,j\}}R_{k,i}\bar{R}_{k,j}f^{i-1}\bar{f}^{j-1})^2}\\
	&=\log\frac{\sum_{i,j,l,m=1}^{n+1}R_{i,j,l,m}(f^{i-1}\bar{f}^{j-1}f^{\prime l-1}\bar{f^\prime}^{m-1}-f^{i-1}\bar{f^\prime}^{j-1}f^{\prime l-1}\bar{f}^{m-1})}{(\sum_{i=1}^{n+1}\sum_{j=1}^{n+1}\sum_{k=1}^{min\{i,j\}}R_{k,i}\bar{R}_{k,j}f^{i-1}\bar{f}^{j-1})^2}\\
\end{aligned}
\nonumber
\end{equation}
Because $f^\prime=\frac{1}{v_0^2}$ and $v$ is total unramified, $f$ is nondegenerate.\\
"IF" part:\par
Assume
$$u_1=\log\frac{\sum_{i,j,l,m=1}^{n+1}R_{i,j,l,m}(f^{i-1}\bar{f}^{j-1}f^{\prime l-1}\bar{f^\prime}^{m-1}-f^{i-1}\bar{f^\prime}^{j-1}f^{\prime l-1}\bar{f}^{m-1})}{(\sum_{i=1}^{n+1}\sum_{j=1}^{n+1}\sum_{k=1}^{min\{i,j\}}R_{k,i}\bar{R}_{k,j}f^{i-1}\bar{f}^{j-1})^2}.$$
Let $A=(r_{i,j})_{n+1\times n+1}\in {\mathcal R}_{n+1},v=[(-f^\prime)^{-\frac{1}{2}}f:(-f^\prime)^{-\frac{1}{2}}]=[f:1]$, and $\phi=\mathcal{A}\circ r_n\circ v\in H_n$, where
\begin{equation}
r_{i,j}=
\begin{cases}
R_{i,j}((i-1)!(n-i+1)!)^{\frac{1}{2}},\forall i\leq j\\
0, \forall i>j
\end{cases}
\nonumber
\end{equation}
Besides, $\phi$ is totally unramified, so we can assume$\{\tilde{u}_i\}$ be the corresponding solution of $\phi$.We can calculate that $\tilde{u}_1=u_1$. That means $\tilde{u}_i=u_i,\,\forall i.$ So $\{u_i\}$ is the corresponding solution of $\phi$.\par
Thus $\{u_i\}$ is a reduced solution.
\end{proof}

\noindent\textbf{Acknowledgement:}
First of all, C.W. would like to give my heartfelt thanks to all the people who have ever helped me with this paper. His sincere and hearty thanks and appreciations go firstly to the professors at USTC, who have given me valuable suggestions for academic studies.
C.W is also extremely grateful to all my friends and classmates who have kindly provided me with assistance and companionship in the course of preparing this paper.
In addition, many thanks go to his family for their unfailing love and unwavering support.
Finally, C.W. is grateful to all those who devote much time to reading this thesis and giving me much advice, which will benefit me in my later study.

Y.S. is supported in part by the National Natural Science Foundation of China (Grant No. 11931009) and Anhui Initiative in Quantum Information Technologies (Grant No. AHY150200).
B.X. would like to express his deep gratitude to  Professor Zhijie Chen at Tsinghua University, Professor Zhaohu Nie at
University of Utah and Professor Guofang Wang at University of Freiburg for their stimulating conversations on Toda systems. Moreover, his research is supported in part by the National Natural Science Foundation of China (Grant Nos. 12271495, 11971450 and 12071449) and the CAS Project for Young Scientists in Basic Research (YSBR-001).\\


\bibliographystyle{amsplain}

\end{document}